\documentclass[conf]{ieeeconf}%

\IEEEoverridecommandlockouts                              %

\usepackage{xr}

\newcommand{\omitted}[1]{}%

\title{%
Optimistic Distributed Submodular Coordination with Applications to Multi-Target Tracking with Multiple Robots
}
\author{Pranjal Sharma, Atharva Navsalkar, Zirui Xu, Vasileios Tzoumas$^\dagger$ \vspace{-1.5em}
	\thanks{%
 $^\dagger$Department of Aerospace Engineering, University of Michigan, Ann Arbor, MI 48109 USA;  {\tt\footnotesize \{spranjal, amnav, ziruixu, vtzoumas\}@umich.edu}} 
}

\usepackage{cite}

\usepackage{comment}
\usepackage{siunitx}
\usepackage{relsize}
\usepackage{ifthen}
\usepackage[colorinlistoftodos]{todonotes}

\usepackage[vlined,ruled,linesnumbered]{algorithm2e}
\usepackage{graphics} %
\usepackage{rotating}
\usepackage{color}
\usepackage{enumerate}
\usepackage[T1]{fontenc}
\usepackage{psfrag}
\usepackage{epsfig} %
\usepackage{booktabs}
\usepackage{graphicx,url}
\usepackage{multirow}
\usepackage{array}
\usepackage{latexsym}
\usepackage{amsfonts}
\usepackage{amsmath}
\usepackage{amssymb}
\usepackage{xstring}
\usepackage{algorithmic}
\usepackage{multirow}
\usepackage{xcolor}
\usepackage{prettyref}
\usepackage{flexisym}
\usepackage{bigdelim}
\usepackage{breqn} %
\usepackage{listings}
\let\labelindent\relax
\usepackage{xspace}
\usepackage{bm}
\graphicspath{{./figures/}}
\usepackage{tikz}
\usetikzlibrary{matrix,calc}
\usepackage{lipsum}
\usepackage{mdwlist}

\let\itemize\undefined
\makecompactlist{itemize}{stditemize}
\usepackage{enumitem}
\usepackage{caption}

\usepackage{amsthm}
\usepackage{mathtools}

\makeatletter
\let\NAT@parse\undefined
\makeatother
\usepackage{hyperref}
\usepackage{cleveref}
\allowdisplaybreaks

\hypersetup{%
  colorlinks=true,
  linkcolor=blue,
  filecolor=magenta,      
  urlcolor=black,
  citecolor=red,
  linkbordercolor={0 0 1}
}

\newtheorem{theorem}{Theorem}
\newtheorem{problem}{Problem}

\newtheorem{lemma}{Lemma}

\newtheorem{definition}{Definition}

\newcommand{\bdmath}{\begin{dmath}}
\newcommand{\edmath}{\end{dmath}}
\newcommand{\beq}{\begin{equation}}
\newcommand{\eeq}{\end{equation}}
\newcommand{\bdm}{\begin{displaymath}}
\newcommand{\edm}{\end{displaymath}}
\newcommand{\bea}{\begin{eqnarray}}
\newcommand{\eea}{\end{eqnarray}}
\newcommand{\beal}{\beq \begin{array}{lll}}
\newcommand{\eeal}{\end{array} \eeq}
\newcommand{\beas}{\begin{eqnarray*}}
\newcommand{\eeas}{\end{eqnarray*}}
\newcommand{\ba}{\begin{array}}
\newcommand{\ea}{\end{array}}
\newcommand{\bit}{\begin{itemize}}
\newcommand{\eit}{\end{itemize}}
\newcommand{\ben}{\begin{enumerate}}
\newcommand{\een}{\end{enumerate}}

\newcommand{\calA}{{\cal A}}
\newcommand{\calB}{{\cal B}}
\newcommand{\calC}{{\cal C}}

\newcommand{\calE}{{\cal E}}

\newcommand{\calG}{{\cal G}}

\newcommand{\calN}{{\cal N}}

\newcommand{\calV}{{\cal V}}

\definecolor{myblue}{RGB}{65 105 225}

\newcommand{\hide}[1]{}

\newcommand{\hiddenText}{{\color{gray} hidden text.}}
\newcommand{\hideWithText}[1]{\hiddenText}

\newcommand{\opt}{^{\star}}

\DeclareRobustCommand{\scenario}[1]{%
  \ifmmode
    \mathsf{#1}%
  \else
    \texorpdfstring{{\fontsize{8.9}{9}\selectfont\sffamily #1}\xspace}{#1}%
  \fi
}

\newcommand{\ie}{\emph{i.e.},\xspace}
\newcommand{\eg}{\emph{e.g.},\xspace}
\newcommand{\myin}{\, \in \,}

\newcommand{\myParagraph}[1]{{\bf #1.}\xspace}

\renewcommand{\opt}{\scenario{OPT}}

\newcommand{\curv}{\kappa}
\newcommand{\ourcurv}{\scenario{coin}}

\newcommand{\alg}{\scenario{O-DOG}}
\newcommand{\dog}{\scenario{DOG}}
\newcommand{\oosg}{\scenario{O-OSG}}

\newcommand{\elem}{v}

\newcommand{\fsf}[1]{\scenario{FSF$^{\star}$}}
\newcommand{\ofsf}[1]{\scenario{O-FSF$^{\star}$}}
\newcommand{\ofsfI}[1]{\scenario{O-FSF$^{\star}|_{#1}$}}

\newcommand{\onalg}{\scenario{OSG}}
\newcommand{\solopt}{\calA^{\opt}}

\PassOptionsToPackage{end}{algorithmic}

\makeatletter
\def\@IEEEfigurecaptionsepspace{\vspace{-5.00pt}}

\def\section{\@startsection{section}{1}{\z@}%
{1.0ex plus 0.3ex minus 0.15ex}%
{0.18ex plus 0.08ex minus 0.08ex}%
{\normalfont\normalsize\centering\scshape}}
\def\subsection{\@startsection{subsection}{2}{\z@}%
{0.7ex plus 0.2ex minus 0.1ex}%
{0.18ex plus 0.08ex minus 0.08ex}%
{\normalfont\normalsize\itshape}}

\def\thm@space@setup{\thm@preskip=1.5pt plus 0.6pt minus 0.4pt%
\thm@postskip=1.5pt plus 0.6pt minus 0.4pt}

\patchcmd{\thebibliography}{\itemsep 0pt plus .5pt\relax}%
{\itemsep -0.90pt plus .15pt\relax}{}%
{\PackageError{layout-spacing}{Bibliography spacing patch failed}{}}
\patchcmd{\thebibliography}%
{\vskip 0.3\baselineskip plus 0.1\baselineskip minus 0.1\baselineskip}%
{\vskip 0.5pt}{}%
{\PackageError{layout-spacing}{Bibliography heading spacing patch failed}{}}
\makeatother

\AtBeginEnvironment{proof}{%
\pretocmd{\trivlist}{\topsep1.5pt plus 0.4pt minus 0.4pt\relax}{}%
{\PackageError{layout-spacing}{Proof spacing hook failed}{}}}

\begin{document}
	
\maketitle
\thispagestyle{empty}
\pagestyle{empty}

\begin{abstract}
We provide a distributed online algorithm for multi-agent submodular coordination in bandit environments, \ie environments whose evolution is unknown a priori and may even adapt to the agents' actions. We are motivated by multi-robot tasks such as target tracking, where the optimal actions change over time and the agents must track them rather than a single fixed optimum. To address such problems, we provide the Optimistic Distributed Online Greedy (\alg) algorithm, that leverages predictions about the environment's behavior and improves regret guarantees. The agents select actions simultaneously over arbitrary, possibly time-varying, network topologies, and each agent acts \emph{optimistically}, \ie leveraging predictions of its upcoming rewards, \eg from forecasts of the targets' motion and of its neighbors' actions. We bound \alg's tracking regret in terms of the agents' prediction error and the bound never exceeds its prediction-free counterpart. Particularly, it is sublinear even when the optimal action changes at every time step, provided the predictions are accurate enough.   
We further quantify the suboptimality due to decentralization as a function of the network structure, and show that accurate predictions accelerate convergence. We validate \alg in simulated multi-robot target tracking.
\end{abstract}

\section{Introduction}\label{sec:Intro}

Teams of robots are increasingly deployed to gather information in dynamic environments: tracking multiple moving targets~\cite{tokekar2014multi,corah2021scalable}, exploring and mapping unknown spaces~\cite{atanasov2015decentralized,corah2019distributed}, and monitoring environmental phenomena~\cite{krause2008near,singh2009efficient}. At every time step $t$, the robots seek to maximize a shared utility via optimization problems of the form
\begin{equation}
\max_{a_{i,t}\in\calV_i,\,\forall i\in\calN}\; f_t(\{a_{i,t}\}_{i\in\calN}),\quad t=1,2,\dots, \label{eq:intro-prob}
\end{equation}
where $\calN$ is the set of robots, $\calV_i$ the actions available to robot~$i$, $a_{i,t}$ the action it selects at time~$t$, and $f_t: \Pi_{i \myin \mathcal{N}}\mathcal{V}_i \mapsto \mathbb{R}$ is the objective function that captures the task utility. Although~\cref{eq:intro-prob} is generally NP-hard~\cite{Feige:1998:TLN:285055.285059}, near-optimal polynomial-time approximation algorithms have been proposed when $f_t$ is \emph{submodular}, a diminishing returns property. For example, the Sequential Greedy (\scenario{SG}) algorithm~\cite{fisher1978analysis}, in which robots choose their actions sequentially to maximize theire marginal contribution, achieves a $1/2$ approximation bound when $f_t$ is submodular. All aforementioned complex tasks can be modeled as submodular coordination problems. Consequently, (\scenario{SG}) and its variants have been commonly adopted in controls, machine learning, and robotics literature~\cite{krause2008near,singh2009efficient,tokekar2014multi,atanasov2015decentralized,gharesifard2017distributed,grimsman2019impact,corah2018distributed,schlotfeldt2021resilient,liu2021distributed,robey2021optimal,rezazadeh2023distributed,konda2022execution,krause2012submodular,xu2025communication}.

Three challenges stand between this approach and deployment. First, $f_t$ is unknown when the robots act: the utility of an action depends on where the targets will be, so the robots must therefore commit to actions before $f_t$ is revealed and learn from the utilities of their actions in hindsight. Second, communication is a bottleneck: \scenario{SG} and its online counterpart~\cite{xu2023online} require each robot to wait for the decisions of all preceding robots in the sequence before it can select an action. Therefore, every action selection round incurs a chain of sequential communication rounds whose length grows quadratically with the team size~\cite{xu2025communication}. Consensus-based simultaneous methods, in turn, require connected, undirected networks and operate on a continuous relaxation of the objective~\cite{zhang2025near}. Third, the environment can be adversarial as targets may deliberately evade the robots, reacting to their motion~\cite{chung2011search}, so that the optimal actions may change at nearly every time step. Additionally, the robots are seldom without foresight: motion models~\cite{kalman1960new} and learned predictors~\cite{rudenko2020human,navsalkar2026} forecast the targets' future positions.  But their accuracy cannot be guaranteed when the targets behave adversarially~\cite{xu2024leveraging}. We therefore seek a bound on the \emph{tracking regret}~\cite{herbster1998tracking}, which compares the robots' cumulative utility against the near-optimal utility that can be attained by the best sequence of time-varying actions chosen with knowledge of the future. Moreover, we seek a bound that improves as the forecasts become more accurate.

To our knowledge, existing methods address at most two of these challenges. Distributed submodular coordination with limited information~\cite{gharesifard2017distributed,grimsman2019impact,konda2022execution,robey2021optimal,xu2025communication} assumes $f_t$ is known, and model-based methods plan on forecasts but presume them correct~\cite{atanasov2015decentralized,schlotfeldt2021resilient}. Online methods handle unknown $f_t$, but some compete only with the best \emph{fixed} action~\cite{streeter2008online,chen2018online}, including \scenario{DOG}~\cite{xu2026distributed}.
\scenario{OSG}~\cite{xu2023online} and its bandit variant~\cite{xu2023bandit} bound the tracking regret but decide sequentially, while \scenario{MA-OSMA}~\cite{zhang2025near} decides simultaneously via consensus on a continuous relaxation but requires undirected connected networks and real-valued messages. None of these use forecasts, and their tracking regret bounds grow linearly in time when the optimal actions change at a constant frequency, \eg when targets actively evade pursuers. Optimistic online learning makes the regret scale with the error of a forecast of the next loss~\cite{chiang2012,rakhlin2013,zhao2024}, but its fast rates in games~\cite{zzlz2022} assume that payoffs are bilinear in the players' mixed strategies. In contrast, we consider problems where each robot's utility depends on its teammates' \emph{realized} actions. In submodular maximization, optimism so far requires a centralized decision maker and has been limited to weighted threshold potentials, a subclass of submodular functions~\cite{sisalem2024online}, while \scenario{MetaBSG}~\cite{xu2024leveraging} accepts or rejects suggested actions as a whole, deciding sequentially.

We propose the \emph{Optimistic Distributed Online Greedy} (\alg)
algorithm for eq.~\eqref{eq:intro-prob} (\Cref{sec:algorithm}) as an extension of \dog~\cite{xu2026distributed}, which
addresses all three challenges.
First, \alg's tracking regret is bounded in terms of the robots'
cumulative forecast error rather than the time horizon (\Cref{sec:bound}). We define the forecast error or accuracy as the cumulative difference between the predicted and true rewards over the time horizon. We prove that accurate forecasts yield a better bound and faster convergence to near optimal actions as compared to the forecast free alternatives. Second, \alg applies to arbitrary, possibly directed and time-varying, communication networks, and its approximation guarantee improves as the network becomes more connected. When the network is strongly connected, it matches the bound with the $1/(1+\kappa_f)$ factor of \scenario{SG}, whereas it degrades down to $1-\kappa_f$ when the robots do not communicate at all (\Cref{thm:odog}). Third, as in \scenario{DOG}, each robot evaluates its candidate actions against the actions of the robots it can receive information from, so the robots decide and communicate simultaneously rather than in a sequential chain.



We evaluate \alg in simulated target tracking under 20 and 50 robot scenarios, where a quarter of the targets are active evaders (\Cref{sec:sim}). With about half of \onalg's messages per round, \alg tracks the targets 27\% (20 robots) and 39\% (50 robots) closer than \oosg, i.e., \onalg equipped with \alg's optimism framework. Moreover, \alg tracks 70\% closer than the published \onalg. This is achieved by replacing \onalg's $N-1$ sequential communication hops per round with a single simultaneous exchange, which lets \alg re-select actions up to 11 times faster.

\section{Distributed Online Submodular Coordination with Bounded Tracking Regret}\label{sec:problem}

We present the problem formulation.  To this end, we use
 the following notation:
\begin{itemize}[leftmargin=*]
    \item $\calV_\calN \triangleq \prod_{i\myin \calN} \,\calV_i$ 
    is the cross product of sets $\{\calV_i\}_{i\myin \calN}$.
    \item $[T]\triangleq\{1,\dots,T\}$ for any positive integer $T$;
    \item $f(\,a\,|\,\calA\,)\triangleq f(\,\calA \cup \{a\}\,)-f(\,\calA\,)$ is the marginal gain of set function $f:2^\calV\mapsto \mathbb{R}$ for adding $a \in \calV$ to $\calA \subseteq\calV$.
    \item $|\calA|$ is the cardinality of a discrete set $\calA$. 
    \item $\calE$ is the set of directed communication edges. 
\end{itemize}
We also use the following framework about  the agents' communication network and their global objective $f$.

\myParagraph{Communication neighborhood}  
When a communication channel exists from agent $j$ to $i$, \ie $(j\rightarrow i) \in \calE$, then $i$ can receive, store, and process information from $j$. The set of all agents that $i$ can receive information from, possibly via relay (\ie multi-hop communication), is denoted by $\calN_{i,t}$. We refer to $\calN_{i,t}$ as agent $i$'s \textit{neighborhood} at time step $t$.

\myParagraph{Communication network} 
The distributed communication network $\calG=\{\calN, \calE\}$ \textit{can be directed and even disconnected}. When the network is fully connected (all agents receive information from all others), we call it \textit{fully centralized}. In contrast, when the network is fully disconnected (all agents are isolated, receiving information from no other agent), we call it \textit{fully decentralized}.

\myParagraph{Optimism}
Before acting at round $t$, each agent $i$ has access to a prediction $\hat{r}_{i,t}$ of the reward $r_{i,t}$ it is about to receive for any action $a_{i,t} \in \mathcal{V}_i$. The predictions can be formed using any information available up to round $t-1$. The predictions can be supplied by any arbitrary heuristic and our guarantees degrade gracefully with the cumulative prediction error.

\begin{definition}[Normalized and Non-Decreasing Submodular Set Function{~\cite{fisher1978analysis}}]\label{def:submodular}
A set function $f:2^\calV\mapsto \mathbb{R}$ is \emph{normalized and non-decreasing submodular} if and only if 
\begin{itemize}[leftmargin=*]
\item (Normalization) $f(\,\emptyset\,)=0$;
\item (Monotonicity) $f(\,\calA\,)\leq f(\,\calB\,)$, $\forall\,\calA\subseteq \calB\subseteq \calV$;
\item (Submodularity) $f(\,s\,|\,\calA\,)\geq f(\,s\,|\,{\mathcal{B}}\,)$, $\forall\,\calA\subseteq {\mathcal{B}}\subseteq\calV$, $s\in \calV$.
\end{itemize}
\end{definition}

Intuitively, let $f(\,\calA\,)$ measure how closely a set $\calA$ of robot actions tracks the targets, \eg by rewarding robots for staying near them. Adding a robot can only bring more targets within reach, so the tracking quality never decreases; this is the non-decreasing property. Moreover, given a fixed number of targets, a robot contributes \emph{less} when its team is larger, since those targets gain little from one more pursuer; this is the submodularity property.

\begin{definition}[2nd-order Submodularity{~\cite{crama1989characterization,foldes2005submodularity}}]\label{def:conditioning}
$f:2^\calV\mapsto \mathbb{R}$ is \emph{2nd-order submodular} if and only if 
\begin{equation}\label{eq:conditioning}
    f(s\,|\,\calC) - f(s\,|\,\calA\cup\calC) \geq f(s\,|\,\calB\cup\calC) - f(s\,|\,\calA\cup\calB\cup\calC),
\end{equation}
for any \emph{disjoint} $\calA, \calB, \calC\subseteq \calV$ and $s\in\calV$.
\end{definition}

Intuitively, if $f(\,\calA\,)$ captures the number of targets tracked by a set $\calA$ of sensors, then \emph{marginal gain of the marginal gains} drops when more sensors are already deployed.

\begin{problem}[Distributed Online Submodular Coordination with Predictions]
\label{pr:dist_submod_max}
At each time  $t\myin [T]$, each agent $i$ selects an action $a_{i,t}$ online, using a prediction $\hat r_{i,t}$ and only  information from and about its neighbors $\calN_{i,t}$, to solve
\begin{equation}
\max_{a_{i,t}\in\calV_i,\;\forall i\in\calN,\;t\in[T]}\; \sum_{t=1}^{T} f_t\big(\{a_{i,t}\}_{i\in\calN}\big),
  \label{eq:prob}
\end{equation}
where each $f_t:{\calV_\calN}\mapsto\mathbb{R}$ is normalized, non-decreasing, and submodular. $f_t$ may be selected adversarially in response to the agents' past actions, and becomes known only after the agents have executed their actions at time step $t$. No assumption is made on the accuracy of the predictions $\hat r_{i,t}$.
\end{problem}

Problem~\ref{pr:dist_submod_max} is a generalization to the problem in~\cite{xu2023online} by considering (i) the availability of reward predictions, and (ii) an arbitrary rather than a connected network. Additionally, Problem~\ref{pr:dist_submod_max} assumes that the agents know the full objective function $f_t$ once they have executed their actions for time step $t$. This assumption is known as the \emph{full information} setting in the literature of online learning~\cite{matsuoka2021tracking}. For target tracking, this means that once the targets have moved, each robot can tell how well any of its actions would have tracked.

The predictions available to each agent might be inaccurate, \eg they might point a robot to the direction opposite to its target. Thus, the algorithm must be resilient against such predictions and match the performance of the corresponding prediction-free algorithm~\cite{xu2023online}. We achieve this by implementing the two-layer structure of the Fixed Share Forecaster$^*$ (\fsf)) algorithm~\cite{matsuoka2021tracking} along with ensuring that any predictions only affect the current action being played and not the future actions (\Cref{subsec:fsf}).


\section{Optimistic Distributed Online Greedy Algorithm (\alg) for Tracking Regret} \label{sec:algorithm}

\setlength{\textfloatsep}{3mm}
\begin{algorithm}[t]
\caption{Optimistic Fixed Share Forecaster$^{\star}$ (\ofsf)).}
\label{alg:ofsf}
\begin{algorithmic}[1]
\REQUIRE Time steps $T$; action set $\mathcal{V}$; and predictions
  $\hat r_t \in [0,1]^{|\mathcal{V}|}$ of the reward vector, each
  available before the step $t$.
\ENSURE Probability distribution
  $p_t \in \{[0,1]^{|\mathcal{V}|} \! : \! \|p_t\|_1 = 1\}$ over the actions
  in $\mathcal{V}$ at each $t = 1,\dots,T$.
 
\STATE $J \leftarrow \lceil \log_2 (2T) \rceil + 1$,\
       $\alpha \leftarrow 3 + \log_e (|\mathcal{V}|\,T)$,\
       $\beta \leftarrow 1/T$;
 
\STATE $\gamma^{(j)} \leftarrow \sqrt{\alpha / T}\ 2^{\,j-1}$
       for all $j = 1,\dots,J$;
 
\STATE $\ell \leftarrow 0$,\
       $\varepsilon_{\ell} \leftarrow \sqrt{\log_e (J)}$,\
       $B_0 \leftarrow 0$;
 
\STATE Initialize $z_1 = [z_{1,1},\dots,z_{J,1}]^{\top}$ with
       $z_{j,1} = 1$, for all $j = 1,\dots,J$;
 
\STATE Initialize
       $w^{(j)}_1 = \big[w^{(j)}_{1,1},\dots,w^{(j)}_{|\mathcal{V}|,1}\big]^{\top}$
       with $w^{(j)}_{i,1} = 1$, \quad for all $j = 1,\dots,J$ and
       $i = 1,\dots,|\mathcal{V}|$;
 
\FOR{each time step $t = 1,\dots,T$}
 
  \STATE \textbf{receive} the prediction $\hat r_t$ of the reward
         vector at the step $t$;
 
  \STATE $\tilde w^{(j)}_t \leftarrow w^{(j)}_t \odot
         \exp\big(\gamma^{(j)} \hat r_t\big)$,\
         $p^{(j)}_t \leftarrow \tilde w^{(j)}_t /
         \big\|\tilde w^{(j)}_t\big\|_1$,
         for all $j = 1,\dots,J$;
 
  \STATE $\mu_{j,t} \leftarrow \hat r_t^{\top} p^{(j)}_t$
         for all $j$,\
         $\tilde z_t \leftarrow z_t \odot
         \exp\big(\varepsilon_{\ell}\, \mu_t\big)$,\
         $q_t \leftarrow \tilde z_t / \|\tilde z_t\|_1$;
 
  \STATE $p_t \leftarrow \sum_{j=1}^{J} q_{j,t}\, p^{(j)}_t$;
 
  \STATE \textbf{observe} the rewards $\{r_{i,t}\}_{i \in \mathcal{V}}$
         that the agent receives by selecting any action
         $i \in \mathcal{V}$ at the step $t$;
 
  \STATE $r_t \leftarrow
         \big[r_{1,t},\dots,r_{|\mathcal{V}|,t}\big]^{\top}$;
 
  \FOR{$j = 1,\dots,J$}
    \STATE $v^{(j)}_{i,t} \leftarrow w^{(j)}_{i,t}
           \exp\big(\gamma^{(j)} r_{i,t}\big)$,
           $i = 1,\dots,|\mathcal{V}|$;
    \STATE $W^{(j)}_t \leftarrow v^{(j)}_{1,t} + \cdots +
           v^{(j)}_{|\mathcal{V}|,t}$;
    \STATE $w^{(j)}_{i,t+1} \leftarrow
           \beta \frac{W^{(j)}_t}{|\mathcal{V}|} +
           (1-\beta)\, v^{(j)}_{i,t}$,
           $i = 1,\dots,|\mathcal{V}|$;
    \STATE $z_{j,t+1} \leftarrow z_{j,t}
           \exp\big(\varepsilon_{\ell}\, r_t^{\top} p^{(j)}_t\big)$;
  \ENDFOR
 
  \STATE $B_t \leftarrow B_{t-1} +
         \big\| r_t^{\top} P_t - \mu_t \big\|_{\infty}^{2}$,\ where
         $r_t^{\top} P_t \triangleq
         \big[r_t^{\top} p^{(1)}_t,\dots,r_t^{\top} p^{(J)}_t\big]^{\top}$;
 
  \IF{$B_t > 4^{\ell}$}
    \STATE $\ell \leftarrow \ell + 1$,\
           $\varepsilon_{\ell} \leftarrow
           \sqrt{\log_e (J) / 4^{\ell}}$,\
           $z_{t+1} \leftarrow [1,\dots,1]^{\top}$;
  \ENDIF
 
\ENDFOR
\end{algorithmic}
\end{algorithm}

We present the Optimistic Distributed Online Greedy algorithm (\alg) for Problem~\ref{pr:dist_submod_max}. \alg leverages, as a subroutine, an algorithm for the problem of \textit{tracking the best expert}~\cite{matsuoka2021tracking}. Particularly, we provide the optimistic extension of the Fixed Share Forecaster$^{\star}$ (\fsf~\cite{matsuoka2021tracking}) subroutine used by \alg. Thus, before presenting \alg in \Cref{subsec:o_dog}, we first discuss the problem of \emph{tracking the best expert} in \Cref{subsec:fsf} and present the proposed subroutine (\ofsf)).

\subsection{The Problem of Tracking the Best Expert}\label{subsec:fsf}

\emph{Tracking the best expert} involves an agent selecting actions to maximize the total reward (utility) across a given number of time steps. However, the challenge is that the reward associated with each action is time-varying and unknown to the agent before the action has been executed. Therefore, to solve the problem, the agent needs to guess the ``best expert'' actions, \emph{i.e.}, the actions achieving the highest reward at each time step.

To state the problem, we use the following notation:
\begin{itemize}[leftmargin=*]
    \item $\mathcal{V}$ denotes the set of actions available to the agent;
    \item $r_{i,t}$ denotes the reward that the agent receives by selecting action $i \myin \mathcal{V}$ at the time step $t$
    \item $r_t \triangleq \left[r_{1,t}, \dots, r_{|\mathcal{V}|,t}\right]^\top$ is the vector of all rewards at $t$;
    \item $i_t^\star \in \arg\max_{i \in \mathcal{V}} r_{i,t}$; \emph{i.e.}, $i_t^\star$ is the index of the ``best expert'' action at time $t$, that is, of the action that achieves the highest reward at time $t$;
    \item $\mathbf{1}(\cdot)$ is the indicator function, \emph{i.e.}, $\mathbf{1}(x) = 1$ if the event $x$ is true, otherwise $\mathbf{1}(x) = 0$.
    \item $P(T) \triangleq \sum_{t=1}^{T-1} \mathbf{1}(i_t^\star \neq i_{t+1}^\star)$; \emph{i.e.}, $P(T)$ counts how many times the best action changes over $T$ time steps.
\end{itemize}

\begin{problem}[Tracking the Best Expert~\cite{cesa2006prediction}]\label{pr:tracking}
Assume a time horizon $H$ of operation discretized to $T$ time steps. The agent selects an action $a_t$ \emph{online} at each time step $t = 1, \dots, T$ to solve the optimization problem
\begin{equation}
    \max_{a_t \in \mathcal{V},\, t = 1, \dots, T} \sum_{t=1}^{T} r_{a_t, t}, \label{eq:tracking_best_expert}
\end{equation}
where $\{r_{i,t}\}_{i \in \mathcal{V}}$ can take any value and become known to the agent \emph{only once the agent has executed action $a_t$}.
\end{problem}

Algorithm~\ref{alg:ofsf}
guarantees a suboptimality bound when solving Problem~\ref{pr:tracking} despite the ``best expert'' actions $i^{\star}_1,\dots,i^{\star}_T$ being unknown a priori to the agent, \emph{i.e.}, despite the challenge that the rewards $\{r_{i,t}\}_{i \in \mathcal{V}}$ become known only \emph{once} the agent has executed action $a_t$. To this end, the algorithm provides the agent with a probability distribution $p_t$ over the action set $\mathcal{V}$ at each time step $t$, from which the agent draws an action $a_t$. Algorithm~\ref{alg:ofsf} additionally makes use of a \emph{prediction}
\begin{equation}
  \hat r_t \triangleq
  \big[\hat r_{1,t},\dots,\hat r_{|\mathcal{V}|,t}\big]^{\top}
  \in [0,1]^{|\mathcal{V}|}
  \label{eq:pred}
\end{equation}
of the reward vector $r_t$ that the agent computes before the time step $t$ from the information it has already observed. No assumption is made on the accuracy of $\hat r_t$: it may be produced by any heuristic, and it enters the guarantee only through the \emph{prediction error budget}
\begin{equation}
  D(T) \triangleq \sum_{t=1}^{T}
  \big\|r_t - \hat r_t\big\|_{\infty}^{2},
  \label{eq:DT}
\end{equation}
which accumulates how far the predictions fall from realized rewards over $T$ time steps. Since $r_t$ and $\hat r_t$ both lie in $[0,1]^{|\mathcal{V}|}$, we have $D(T) \leq T$. Thus, the prediction error budget is bounded by $T$ regardless the prediction accuracy.

\medskip
\noindent\textbf{Intuition Behind Algorithm~\ref{alg:ofsf}.} \ofsf  rruns $J$ base learners in parallel. Base learner $j$ keeps a weight $w^{(j)}_{i,t}$ for each action $i$, which grows the more reward action $i$ has received relative to the others (lines~14--16). The learning rate $\gamma^{(j)}$ sets how strongly each new reward moves these weights, so larger rates adapt faster to a changing environment. The prediction enters only in the distribution the agent plays (lines~8--10): at line~8, the weights are tilted by $\hat r_t$, as though $\hat r_t$ had already been observed. Since line~14 updates the weights $w^{(j)}_{i,t}$ unaffected by $\hat r_t$, a poor prediction only affects the round at which it was issued, which is why $D(T)$ sums per-round errors.

Large rates suit environments whose best action changes often (large $P(T)$) and predictably (small $D(T)$), which make it safe to react quickly. Small rates suit the opposite, making the algorithm resilient against harmful predictions.

\smallskip
\noindent\textbf{Adapting to Unknown $P(T)$ and $D(T)$.} A meta learner combines the base learners into $p_t$, itself tilted by $\hat r_t$ (lines~9--10). This is the distribution that is played by algorithm every round. Since $P(T)$ is never observed, the meta hedges across all rates; since $D(T)$ is, the meta's own rate shrinks as the observed prediction error grows (lines~18--20) via a doubling trick~\cite{cesa2006prediction}. This two-layer structure of $J$ base learners combined with a meta learner allows \ofsf{} to operate in adversarial environments and converge towards a near-optimal sequence of actions.

\setlength{\textfloatsep}{3mm}
\begin{algorithm}[t]
\caption{Optimistic Distributed Online Greedy (\alg) for agent $i$.}
\label{alg:odog}
\begin{algorithmic}[1]
\REQUIRE Time steps $T$; agent $i$'s action set $\mathcal{V}_i$; and
  agent $i$'s neighborhood $\mathcal{N}_i$.
\ENSURE Agent $i$'s action $a^{\alg}_{i,t}$ at each $t = 1,\dots,T$.
 
\STATE Initialize an \ofsf with the value of the parameters $T$ and
  $\mathcal{V}_i$, denoted by \ofsfI{i};
 
\FOR{each time step $t = 1,\dots,T$}
 
  \STATE \textbf{form} the prediction
    $\hat r^{(i)}_t \in [0,1]^{|\mathcal{V}_i|}$ of the reward vector at
    step $t$, using information observed up to step $t-1$;
 
  \STATE \textbf{input} $\hat r^{(i)}_t$ to \ofsfI{i}
    (per line~7 of Algorithm~\ref{alg:ofsf});
 
  \STATE \textbf{get} the output $p^{(i)}_t$ from \ofsfI{i};
 
  \STATE \textbf{draw} an action $a^{\alg}_{i,t}$ from the distribution
    $p^{(i)}_t$, and \textbf{execute} it;
 
  \STATE \textbf{send} $a^{\alg}_{i,t}$ to the out-neighbors of $i$, and
    \textbf{relay} the actions received from other agents, until
    $\mathcal{A}^{\alg}_{\mathcal{N}_i,t} \triangleq
    \{a^{\alg}_{j,t}\}_{j \in \mathcal{N}_i}$ has been received;
 
  \STATE \textbf{observe} the objective function
    $f_t :\mathcal{V}_{\mathcal{N}} \mapsto \mathbb{R}$;
 
  \FOR{every $a \in \mathcal{V}_i$}
    \STATE $r^{(i)}_{a,t} \leftarrow
      f_t\big(a \mid \mathcal{A}^{\alg}_{\mathcal{N}_i,t}\big)$;
  \ENDFOR
 
  \STATE $r^{(i)}_t \leftarrow \{r^{(i)}_{a,t}\}_{a \in \mathcal{V}_i}$;
 
  \STATE \textbf{input} $r^{(i)}_t$ to \ofsfI{i}
    (per line~11 of Algorithm~\ref{alg:ofsf});
 
\ENDFOR
\end{algorithmic}
\end{algorithm}

\subsection{\alg Algorithm}\label{subsec:o_dog}
We now introduce the Optimistic Distributed Online Greedy (\alg) algorithm (\Cref{alg:odog}). \alg retains the simultaneous decision structure of Distributed Online Greedy (\dog)~\cite{xu2026distributed} and equips each agent with an optimistic tracking learner, extending \dog's guarantees from static regret to the tracking regret of Problem~\ref{pr:dist_submod_max} by leveraging at the agent level \ofsf ((\Cref{alg:ofsf}).


In \alg, agent $i$ instead scores its actions' marginal gain given the actions of the its local communication neighborhood $\mathcal{N}_{i,t}$,
\begin{equation}
  r^{(i)}_{a,t} \;=\;
  f_t\big(a \,\big|\, \mathcal{A}^{\alg}_{\mathcal{N}_{i,t},t}\big),
  \qquad a \in \mathcal{V}_i ,
  \label{eq:odogreward}
\end{equation}
which involves no ordering of the agents. All agents evaluate eq.~\eqref{eq:odogreward} at the same time, as soon as their neighbors' actions have been received, enabling simultaneous decision making.

In more detail, \alg starts by instructing each agent $i \in \mathcal{N}$ to initialize an \ofsf{} --- we denote the \ofsf{} onboard each agent $i$ by \ofsfI{i}. Specifically, agent $i$ initializes \ofsfI{i} with its action set $\mathcal{V}_i$ and with the number $T$ of total time steps (line~1). Then, at each time step $t = 1,\dots,T$, all agents in parallel:
\begin{itemize}[leftmargin=*]
  \item Each agent $i$ forms a prediction $\hat r^{(i)}_t$ of its reward vector from the information it has observed up to $t-1$, inputs it to \ofsfI{i} per line~7 of Algorithm~\ref{alg:ofsf}, and draws an action $a^{\text{\small \alg}}_{i,t}$ given the probability distribution $p^{(i)}_t$ output by \ofsfI{i} (lines~3--6).
  \item All agents execute their actions $\{a^{\text{\small \alg}}_{i,t}\}_{i \in \mathcal{N}}$ simultaneously and broadcast them, so that each agent $i$ receives the neighbors' actions $\mathcal{A}^{\text{\small \alg}}_{\mathcal{N}_{i,t},t}$ within the time step $t$. The agents then observe $f_t : 2^{\mathcal{V}_{\mathcal{N}}} \mapsto \mathbb{R}$ (lines~7--8).
  \item Each agent $i$ computes the marginal gain in eq.~\eqref{eq:odogreward} of each of its actions $a \in \mathcal{V}_i$ (lines~9--12). No agent waits on any other as the computation requires only the neighbors' actions, which have already been received.
  \item Finally, each agent $i$ makes the vector of its computed marginal gains observable to \ofsfI{i} per line~11 of Algorithm~\ref{alg:ofsf} (line~13). With this input, \ofsfI{i} will compute $p^{(i)}_{t+1}$, \emph{i.e.}, the probability distribution over the agent $i$'s actions for the next time step $t+1$.
\end{itemize}

\section{Performance Guarantees of \alg}\label{sec:bound}

We bound \alg's suboptimality with respect to the optimal actions the agents would select if they knew $\{f_t\}_{t=1,\dots,T}$ a priori. Particularly, we bound \alg's \emph{tracking regret}, proving that it gracefully degrades with the environment's capacity to select $\{f_t\}_{t=1,\dots,T}$ adversarially, and \emph{improves} with the accuracy of the agents' predictions. 

We proceed in two steps. First, we bound the tracking regret of \ofsf~at the level of a single agent (Theorem~\ref{thm:ofsf}). Then, we bound the tracking regret of the network (Theorem~\ref{thm:odog}), which requires quantifying the environment's capacity to select $\{f_t\}_{t=1,\dots,T}$ adversarially and the agents' capacity to predict it. To these ends, we use the notation:
\begin{itemize}[leftmargin=*]
  \item $\mathcal{A}^{\mathrm{OPT}}_t \in
     \arg\max_{a_{i,t} \in \mathcal{V}_i,\,\forall i \in \mathcal{N}} f_t(\{a_{i,t}\}_{i \in \mathcal{N}})$; \emph{i.e.}, $\mathcal{A}^{\mathrm{OPT}}_t$ is the set of optimal actions the agents would select at time step $t$ if they knew the $f_t$ a priori;
  \item $a^{\mathrm{OPT}}_{i,t}$ is agent $i$'s action among the actions in $\mathcal{A}^{\mathrm{OPT}}_t$;
  \item $\mathcal{A}_t \triangleq \{a_{i,t}\}_{i \in \mathcal{N}}$; \emph{i.e.}, $\mathcal{A}_t$ is the set of all agents' actions at time step $t$.
\end{itemize}

\begin{definition}[$\alpha$-Approximate Tracking Regret]
\label{def:treg}
Consider any sequence of action sets $ \{\mathcal{A}_t\}_{t=1,\dots,T}$, and let $\alpha \in (0,1]$. Then, $\{\mathcal{A}_t\}_{t=1,\dots,T}$'s $\alpha$-approximate tracking regret is
\begin{equation}
  \mathrm{TReg}^{(\alpha)}_T\big(\{\mathcal{A}_t\}_{t=1,\dots,T}\big) \triangleq \alpha \sum_{t=1}^{T} f_t\big(\mathcal{A}^{\mathrm{OPT}}_t\big)
  - \sum_{t=1}^{T} f_t\big(\mathcal{A}_t\big).
  \label{eq:treg}
\end{equation}
\end{definition}

\Cref{eq:treg} evaluates $\{\mathcal{A}_t\}_{t=1,\dots,T}$'s suboptimality with respect to the optimal actions $\{\mathcal{A}^{\mathrm{OPT}}_t\}_{t=1,\dots,T}$ the agents would select if they knew $\{f_t\}_{t=1,\dots,T}$ a priori. The optimal value value $\sum_{t=1}^{T} f_t\big(\mathcal{A}^{\mathrm{OPT}}_t\big)$ is discounted by a factor of $\alpha$ in \cref{eq:treg} since solving Problem~\ref{pr:dist_submod_max} exactly is NP-hard even when $\{f_t\}_{t=1,\dots,T}$ are known a priori~\cite{sviridenko2017optimal}. Specifically, the best possible approximation bound in polynomial time is $1 - 1/e$ \cite{sviridenko2017optimal}, while Sequential Greedy and \onalg achieve the near-optimal bound $1/2$\cite{xu2023online}. In this paper, we establish $\alpha = 1/(1+\kappa_f)$, where $\kappa_f \in [0,1]$ is the curvature of $f$ (\Cref{def:curvature}). This is the curvature-dependent guarantee equals $1/2$ in the worst case $\kappa_f = 1$. In contrast to Sequential Greedy and \onalg, however, \alg attains it without ordering the agents.

\begin{definition}[Environment's Total Adversarial Effect]
\label{def:adv}
The environment's total adversarial effect over $T$ time steps is
\begin{equation}
  \Delta(T) \triangleq \sum_{t=1}^{T-1} \sum_{i \in \mathcal{N}} \mathbf{1}\big(a^{\mathrm{OPT}}_{i,t} \neq a^{\mathrm{OPT}}_{i,t+1}\big).
  \label{eq:adv}
\end{equation}
\end{definition}$\Delta(T)$ counts how many times the optimal actions of the agents must shift across the $T$ steps to adapt to the changing $f_{(\cdot)}$. The larger it is, the more frequently the agents need to change actions to remain optimal.
 
\begin{definition}[Network Prediction Error]
\label{def:prederr}
Agent $i$'s prediction error over $T$ time steps is $D_i(T) \triangleq \sum_{t=1}^{T}\|r^{(i)}_t - \hat r^{(i)}_t\|_\infty^2$, and the network's prediction error is
\begin{equation}
  \mathcal{D}(T) \triangleq \sum_{i \in \mathcal{N}} \mathbb{E}\big[\max\{D_i(T),1\}\big].
  \label{eq:netpred}
\end{equation}
\end{definition}$\mathcal{D}(T)$ is the network counterpart of ${D}(T)$ in~\cref{eq:DT}. Since $r^{(i)}_t,\hat r^{(i)}_t \in [0,1]^{|\mathcal{V}_i|}$, we have $\mathcal{D}(T) \leq |\mathcal{N}|\,T$ regardless of prediction accuracy.

\begin{theorem}[Agent-Level Approximation Performance]
\label{thm:ofsf}
Consider \ofsf ((\Cref{alg:ofsf}) run on the reward sequence $\{r_t\}_{t=1}^{T} \subseteq [0,1]^{|\mathcal{V}|}$ with predictions $\{\hat r_t\}_{t=1}^{T} \subseteq [0,1]^{|\mathcal{V}|}$, and let $\{i^\star_t\}_{t=1}^{T}$ be any comparator sequence that changes value at most $P(T)$ times. Then,
\begin{equation}
  \sum_{t=1}^{T} r_{i^\star_t,t} - \sum_{t=1}^{T} r_t^{\top} p_t \leq \tilde{O}\Big[\sqrt{\big(1+P(T)\big)\,\bar{D}(T)}\Big],
  \label{eq:thmofsf}
\end{equation}
where $\bar{D}(T) = \max\{D(T),1\}$ and $\tilde{O}[\cdot]$ hides log terms. The bound holds for every realization, with no assumption on how the rewards or predictions are generated.
\end{theorem}
 
Theorem~\ref{thm:ofsf} is in order the same as \cite[Corollary~1]{matsuoka2021tracking} with the horizon $T$ replaced by the prediction error $\bar{D}(T)$. \ofsf provides a $\tilde O[\sqrt{(1+P(T))\bar{D}(T)}]$ regret against the prediction-free $\tilde O[\sqrt{(1+P(T))T}]$ regret of \fsf{}. Since $\bar{D}(T) \leq T$, it is never weaker. The proof is in Appendix~\ref{app:ofsf}.

We also need the following definitions to present the approximation performance of \alg.

\begin{definition}[Centralization of Information, adapted from~\cite{xu2025communication}]
\label{def:coin}
For each time step $t\in[T]$, consider a function $f_t:\calV_\calN\mapsto\mathbb{R}$ and communication neighborhoods $\{\calN_{i,t}\}_{i\myin\calN}$, where each agent $i\in\calN$ has selected an action $a_{i,t}$, and let $\calN_{i,t}^c\triangleq\calN\setminus(\calN_{i,t}\cup\{i\})$. Then, at time $t$, agent $i$'s \emph{centralization of information} is defined as
\begin{equation}\label{eq:ourcurv}
  \ourcurv_{f_t,i}(\calN_{i,t})\triangleq f_t(a_{i,t}) - f_t(a_{i,t}\,|\,\{a_{j,t}\}_{j\myin\calN_{i,t}^c}).
\end{equation}
\end{definition}
To capture the suboptimality cost of decentralization, $\ourcurv_{f_t,i}$ measures how much $a_{i,t}$ can overlap with the actions of agents that are not agent $i$'s neighbors at time $t$. In the best scenario, where $a_{i,t}$ does \underline{not} overlap with other actions at all, \ie $f_t(a_{i,t}\,|\,\{a_{j,t}\}_{j\myin\calN_{i,t}^c})=f_t(a_{i,t})$, then
$\ourcurv_{f_t,i}=0$. In the worst case instead where $a_{i,t}$ is fully redundant, \ie $f_t(a_{i,t}\,|\,\{a_{j,t}\}_{j\myin\calN_{i,t}^c})=0$, then $\ourcurv_{f_t,i}=f_t(a_{i,t})$.

\begin{definition}[Curvature~\cite{conforti1984submodular}]\label{def:curvature}
        The curvature of a normalized submodular function $f\colon 2^{\calV}\mapsto \mathbb{R}$ is defined as
        \begin{equation}
            \kappa_f\triangleq 1-\min_{\elem\in\calV}{[f(\calV)-f(\calV\setminus\{\elem\})]}/{f(\elem)}.
        \end{equation}
\end{definition}$\kappa_f$ measures how far~$f$ is from modularity: if $\kappa_f=0$, then  $f(\calV)-f(\calV\setminus\{v\})=f(v)$, $\forall v\in\calV$, \ie $f$ is modular. In~contrast, $\kappa_f=1$ in the extreme case where there exist $v\in\calV$ such that $f(\calV)=f(\calV\setminus\{v\})$, \ie $v$~has no contribution in the presence of $\calV\setminus\{v\}$.

\begin{theorem}[Network-Level Approximation Performance]\label{thm:odog}
Over $t\in [T]$, given the communication neighborhoods $\{\calN_{i,t}\}_{i\in\calN}$, \alg instructs each agent $i\in\calN$ to select actions $\{a_{i,t}\}_{t\in[T]}$ that guarantee
\begin{itemize}[leftmargin=3.5mm]
    \item If the network is \textbf{fully centralized}, \ie $\calN_{i,t}=\calN\setminus\{i\}$ for all $t\in[T]$,
    {\small\begin{equation}\label{eq:thm-1}
        \mathbb{E}\Big[\frac{1}{T}\sum_{t=1}^{T}f_t(\calA_t)\Big] \geq
        \frac{1}{1+\curv_f}\,\mathbb{E}\Big[\frac{1}{T}\sum_{t=1}^{T}f_t(\calA^\opt_t)\Big]
        - \phi(T).
    \end{equation}}
    \item If the network is \textbf{fully decentralized}, \ie $\calN_{i,t}=\emptyset$ for all $t\in[T]$,
    {\small\begin{equation}\label{eq:thm-2}
        \mathbb{E}\Big[\frac{1}{T}\sum_{t=1}^{T}f_t(\calA_t)\Big] \geq
        (1-\curv_f)\,\mathbb{E}\Big[\frac{1}{T}\sum_{t=1}^{T}f_t(\calA^\opt_t)\Big]
        - \chi(T).
    \end{equation}}
    \item If the network is \textbf{anything in between} fully centralized and fully decentralized, \ie $\calN_{i,t}\subseteq\calN\setminus\{i\}$, and $f_t$ is 2nd-order submodular (\Cref{def:conditioning}),
    {\small\begin{align}
    &\mathbb{E}\Big[\frac{1}{T}\sum_{t=1}^{T}f_t(\calA_t)\Big] \geq \frac{1}{1+\curv_f}\,\mathbb{E}\Big[\frac{1}{T}\sum_{t=1}^{T}f_t(\calA^\opt_t)\Big] \nonumber\\
    &- \frac{\kappa_f}{1+\curv_f}\mathbb{E}\Big[\frac{1}{T}\sum_{t=1}^{T}\sum_{i\in\calN}
  \ourcurv_{f_t,i}(\calN_{i,t})\Big]
    - \psi(T).
      \label{eq:thm-4}
    \end{align}}
    where {\small $\phi(T),\chi(T),\psi(T) \! = \! \tilde{O}\left( \! \sqrt{\mathcal{D}(T)\, \! \big[\mathbb{E}[\Delta(T)]\!+\!|\mathcal{N}|\big]}\,\big/\,T\right)$}, $\kappa_{f}\triangleq\max_{t\in [T]}\kappa_{f_t}$, and $\Delta(T)$ and $\mathcal{D}(T)$ are given in \Cref{def:adv,def:prederr}.
\end{itemize}
Particularly, the expectation is due to \alg's internal randomness and the environment's, and $\tilde{O}(\cdot)$ hides $\log$~terms.
\end{theorem}
 
 
\Cref{thm:odog} is the counterpart of \cite[Theorem~1]{xu2026distributed} with the static optimum $\solopt$ replaced by the time-varying optimum $\calA^\opt_t$. Relative to the tracking-regret bound
$\tilde O\{\sqrt{|\calN|T[\Delta(T)+|\calN|]}\}$ of \onalg~\cite[Theorem~1]{xu2023online}, the product $|\calN|\,T$ is replaced by $\mathcal{D}(T)\le|\calN|\,T$.

The role of prediction is to make $\phi(T)$, $\chi(T)$ and $\psi(T)$ vanish in regimes where they would otherwise not. Without prediction, $\mathcal{D}(T)=\Theta(|\calN|T)$ and the rate vanishes only if $\Delta(T)$ grows sublinearly in $T$, which fails when the environment reacts to the agents, for example when targets evade the robots tracking them. Since $\Delta(T)\le|\calN|(T-1)$ unconditionally, the rate vanishes whenever $\mathcal{D}(T)=\tilde o(T/|\calN|)$, \emph{regardless of how $\Delta(T)$ grows}: an environment that changes frequently but predictably is trackable.

   

The bound in \cref{eq:thm-4} depends on $\scenario{coin}_{f_t,i}$ capturing the suboptimality due to decentralization. \cref{eq:thm-1,eq:thm-2,eq:thm-4} imply \alg's suboptimality will improve if the agents have larger multi-hop coordination neighborhoods. Importantly, the $1/(1+\curv_f)$ suboptimality bound with a fully connected network recovers the bound in~\cite{conforti1984submodular} and is near-optimal as the best possible bound for~\cref{eq:prob} is $1-\kappa_f/e$~\cite{sviridenko2017optimal}.\footnote{{The bounds $1/(1+\curv_f)$ and $1-\kappa_f/e$ become $1/2$ and $1-1/e$ when, in the worst case, $\kappa_f=1$.}}

The proof of \Cref{thm:odog} follows \cite[Theorem~1]{xu2026distributed}, bounding the network's tracking regret by the agents' individual tracking regrets plus a decentralization term, with two changes: the comparator $\calA^\opt_t$ is time-varying, and each agent's regret is bounded by \Cref{thm:ofsf}, 
valid for every realization of the environment and neighbor action sequence. Summing over agents via Cauchy--Schwarz yields the rates.

\section{Numerical Evaluation}
\label{sec:sim}

We compare \alg with \onalg~\cite{xu2023online} and its optimistic version \oosg, which uses \ofsf{} as the agent-level subroutine instead of \fsf{}. We scale the multi-robot target-tracking scenario in~\cite{xu2023online} to $N=20$ robots and 16 targets in a 300\,m radius arena, and to $N=50$ robots and 40 targets in a 474\,m radius arena. The specific radius for $N=50$ was chosen to simulate a similar density as the 20 agent scenario.

\textbf{Setup.} The robots observe the targets' current positions but not their future motion. Each robot is a unicycle whose action is one of 31 (speed, turn-rate) commands. Speeds are chosen from $\{0,4,8,12,16\}$\,m/s, turn rates from $\{0,\pm0.35,\pm1,\pm2\}$\,rad/s, and we limit $v|\omega|\le20$\,m/s$^2$. An agent holds its chosen action for one decision period $P$. The world evolves in steps of $0.1$\,s over a $120$\,s episode, so each learner plays $T=120\,\mathrm{s}/P$ rounds. Half of the targets move \emph{randomly} (a new heading and a speed in $[4,12]$\,m/s every 2\,s on average), a quarter follow \emph{polylines} with 1--4 turns, and a quarter are \emph{adversarial}: they move at 8\,m/s by default, but evade at 12\,m/s for 6\,s when a robot comes within 20\,m. This evasion is in the direction that most reduces the team's coverage of them. Their evasion has a 6\,s cooldown, as in~\cite{xu2023online}. The robots maximize
\begin{equation}
  f_t(\mathcal A) = \sum_{q}\Big[\, d_{\max} - \Big(\textstyle\sum_{i\in\mathcal C_{q,t}} \|x_{i,t}-y_{q,t}\|^{-1} + d_{\max}^{-1}\Big)^{-1} \Big], \label{eq:coverage}
\end{equation}
where $\mathcal C_{q,t}=\{i\in\mathcal N:\|x_{i,t}-y_{q,t}\|\le r\}$ is the set of robots within the coverage radius $r=150$\,m of target $q$. A robot outside $\mathcal C_{q,t}$ does not contribute to target $q$, while $f_t$ remains monotone submodular. \Cref{eq:coverage} is a bounded variant of the objective of~\cite{xu2023online}, with $d_{\max}=300$\,m. A robot's reward is its marginal gain on the targets' realized positions during the period, conditioned on its teammates' actions. \alg and \oosg predict this reward by a constant-turn-rate extrapolation of each target. For computing prediction of marginal gains, \alg uses each neighbor's most probable action under its previous-round weights (included in communication), whereas \oosg uses its predecessors' actual choices. The radius $r$ applies to all algorithms.

\textbf{Communication-limited decision periods.} In \onalg and \oosg, robot $i$ receives the actions of its $i-1$ predecessors over $N-1$ sequential hand-offs, i.e., $N(N-1)/2$ messages per round. In \alg, each robot exchanges realized and most probable actions (given the robot's action distribution at the time) with its $k$ nearest teammates in a single simultaneous exchange, i.e., $Nk$ messages per round. Following the decision-time accounting of~\cite{xu2025communication}, a round takes
\begin{align*}
  \tau_{\textsc{o-osg}} &= N t_c + (N-1)\,\tau + \textstyle\sum_{m=1}^{N-1}\lceil m/c\rceil\, S/B, \\
  \tau_{\textsc{o-dog}} &= t_c + \tau + \lceil k/c\rceil\, S/B,
\end{align*}
where $t_c=0.50$\,ms is a robot's computation per round (measured in our simulator), $\tau=50$\,ms is the per-hop latency (measured short-range Zigbee round trips stay below 100\,ms~\cite{haque2022zigbee}), $S=100$\,B is the size one agent's single round communication, $c=5$ messages can be exchanged simultaneously, and $B$ is the link rate: $0.25$\,Mbps (Digi XBee~3 Zigbee) or $100$\,Mbps (Silvus SL5200), as in~\cite{xu2025communication}. Each algorithm's period $P$ is the smallest multiple of $0.1$\,s that is at least its round time $\tau_{\textsc{alg}}$ (\Cref{tab:comm}). \alg's round is at most 64\,ms for every $k\le20$, so every period we test is feasible. We run each $k$ at its best measured period among $\{0.2,0.3,0.5\}$\,s (0.3\,s for all $k\ge5$). \onalg and \oosg run at their shortest feasible period as longer periods only degrade tracking at those time scales. We report the mean distance from each target to its nearest robot over 20 random initializations.

\begin{table}[t]
  \centering
  \caption{\small Round time, decision period $P$, and mean target-to-nearest-robot distance (m; mean over 20 initializations) for 20 and 50 robots at two link rates. \alg runs at $P=0.3$\,s for both $k$.}
  \label{tab:comm}
  \footnotesize
  \setlength{\tabcolsep}{4pt}
  \begin{tabular}{lcccc}
    \toprule
     & \multicolumn{2}{c}{$N=20$} & \multicolumn{2}{c}{$N=50$} \\
    \cmidrule(lr){2-3}\cmidrule(lr){4-5}
     & 0.25\,Mbps & 100\,Mbps & 0.25\,Mbps & 100\,Mbps \\
    \midrule
    OSG/O-OSG ($P$) & 1.11\,s (1.2) & 0.96\,s (1.0) & 3.32\,s (3.4) & 2.48\,s (2.5) \\
    O-DOG ($P$), & $\le$64\,ms & 51\,ms & $\le$64\,ms & 51\,ms \\
    \midrule
    OSG & 63.5 & 60.0 & 59.9 & 60.7 \\
    O-OSG & 30.9 & 29.9 & 43.4 & 38.3 \\
    O-DOG, $k=5$ & 22.5 & 22.5 & 26.3 & 26.3 \\
    O-DOG, $k=10$ & \textbf{16.6} & \textbf{16.6} & \textbf{19.1} & \textbf{19.1} \\
    \bottomrule
  \end{tabular}
\end{table}


\begin{figure*}[t]
  \centering
  \includegraphics[width=\textwidth]{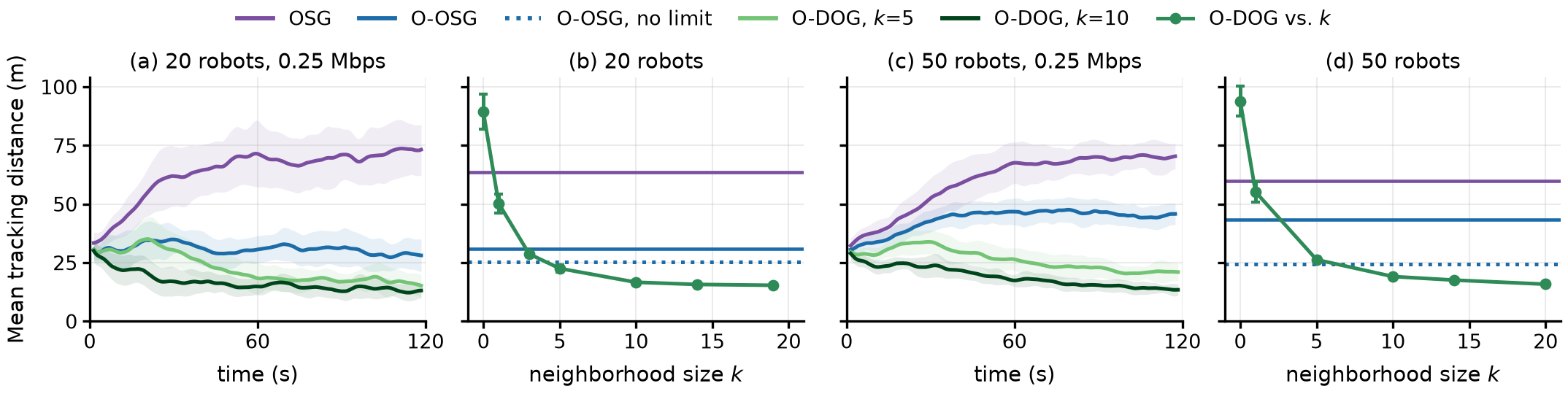}
  \caption{\small (a, c) Mean tracking distance with 20 and 50 robots at 0.25\,Mbps (3\,s moving average; shaded: $\pm1$ standard deviation over 20 seeds); \onalg and \oosg decide every 1.2\,s (a) and 3.4\,s (c), \alg every 0.3\,s. (b, d) Episode mean versus \alg's neighborhood size $k$ ($Nk$ messages per round), each $k$ at its best period (error bars: 95\% confidence intervals). Solid horizontal lines: \onalg and \oosg at their 0.25\,Mbps periods ($N(N-1)/2$ messages per round); dotted: \oosg without communication limits ($P=0.5$\,s).}
  \label{fig:results}
\end{figure*}

\textbf{Results.} \alg's learner and predictions help even under sequential coordination: \oosg tracks 51\% closer than \onalg with 20 robots and 28\% closer with 50 (\Cref{tab:comm}). However, sequential communication limits how often \oosg can decide. Each of its $N-1$ hand-offs costs one link latency, so its round grows with the team: it re-plans only a quarter as often as \alg with 20 robots and less than a tenth as often with 50, whereas \alg's single simultaneous exchange does not depend on $N$. Since the targets keep moving between decisions, the slower loop costs accuracy: with about half of \oosg's messages, \alg with $k=5$ tracks 27\% closer with 20 robots and 39\% closer with 50 at 0.25\,Mbps (\Cref{fig:results}a,c). A faster link barely helps \oosg because its round is dominated by latency rather than bandwidth. With $k\ge10$, \alg is even better than \oosg without any communication limit (by up to 34\%). This is because \alg allows every robot to condition on its nearest teammates, whereas in \oosg the ordering is fixed and the robots earlier in the sequence decide without knowing the choices of those after them.

\Cref{fig:results}b,d show how much coordination \alg needs. Without neighbors ($k=0$), the robots ignore each other and track worse than \onalg; a single neighbor already beats \onalg with 20 robots and roughly matches it with 50, and from $k=5$ onward \alg beats \oosg on every initialization at both link rates. Beyond $k\approx10$ the gains are marginal, which is consistent with the locality of \cref{eq:coverage}. Any robot's reward depends only on teammates within the coverage radius $r$, so additional, more distant neighbors add little value. %

\section{Conclusion} \label{sec:con}
We presented the Optimistic Distributed Online Greedy (\alg) algorithm for multi-agent submodular coordination in unpredictable environments. Leveraging tools from optimistic online learning, \alg enables agents to make simultaneous online decisions over arbitrary, possibly time-varying, communication networks, while exploiting predictions of their upcoming rewards. We bounded \alg's tracking regret in terms of the agents' prediction error rather than the time horizon: the bound never exceeds its prediction-free counterpart, and remains sublinear even when the optimal actions change at every time step, provided the environment changes predictably. In simulated multi-robot target tracking with a 20 agent network, \alg tracked the targets by up to 27\% closer than the state-of-the-art sequential algorithm \onalg equipped with optimism while transmitting approximately half the number of messages. On a 50 agent network that gap grew to 40\%. Compared to \onalg without any optimism, \alg tracks up to 70\% closer, thereby achieving a scalable distributed submodular coordination architecture for tracking regret.

\textbf{Future work.} We will extend \alg's guarantees to actions whose rewards are revealed only over multiple time steps, such as multi-step trajectories executed in the style of receding horizon, by using delayed feedback techniques~\cite{sharma2026async}. We will also extend \alg to bandit feedback, where each agent observes the reward of its selected action only.

\bibliographystyle{IEEEtran}
\bibliography{references}

@Article{haque2022zigbee,
AUTHOR = {Haque, Khandaker Foysal and Abdelgawad, Ahmed and Yelamarthi, Kumar},
TITLE = {Comprehensive Performance Analysis of Zigbee Communication: An Experimental Approach with XBee S2C Module},
JOURNAL = {Sensors},
VOLUME = {22},
YEAR = {2022},
NUMBER = {9},
ARTICLE-NUMBER = {3245},
URL = {https://www.mdpi.com/1424-8220/22/9/3245},
PubMedID = {35590935},
ISSN = {1424-8220},
DOI = {10.3390/s22093245}
}

@misc{sharma2026async,
      title={Asynchronous Distributed Bandit Submodular Maximization under Heterogeneous Communication Delays}, 
      author={Pranjal Sharma and Zirui Xu and Vasileios Tzoumas},
      year={2026},
      eprint={2604.06430},
      archivePrefix={arXiv},
      primaryClass={eess.SY},
      url={https://arxiv.org/abs/2604.06430}, 
}

@inproceedings{sisalem2024online, author = {Salem, Tareq Si and {\"O}zean, G{\"o}zde and Nikolaou, Iasonas and Terzi, Evimaria and Ioannidis, Stratis}, title = {Online submodular maximization via online convex optimization}, year = {2024}, isbn = {978-1-57735-887-9}, publisher = {AAAI Press}, url = {https://doi.org/10.1609/aaai.v38i13.29425}, doi = {10.1609/aaai.v38i13.29425}, articleno = {1677}, numpages = {9}, series = {AAAI'24/IAAI'24/EAAI'24} }

@misc{navsalkar2026,
      title={Self-Adaptive Learning and Model Predictive Control for Tracking Unknown Dynamics with No Regret}, 
      author={Atharva Navsalkar and Hongyu Zhou and Vasileios Tzoumas},
      year={2026},
      eprint={2607.26370},
      archivePrefix={arXiv},
      primaryClass={cs.RO},
      url={https://arxiv.org/abs/2607.26370}, 
}

@article{chung2011search,
  author  = {Chung, Timothy H. and Hollinger, Geoffrey A. and Isler, Volkan},
  title   = {Search and pursuit-evasion in mobile robotics: A survey},
  journal = {Autonomous Robots},
  volume  = {31}, number = {4}, pages = {299--316}, year = {2011},
  doi     = {10.1007/s10514-011-9241-4}
}

@article{rudenko2020human,
  author  = {Rudenko, Andrey and Palmieri, Luigi and Herman, Michael and
             Kitani, Kris M. and Gavrila, Dariu M. and Arras, Kai O.},
  title   = {Human motion trajectory prediction: A survey},
  journal = {The International Journal of Robotics Research},
  volume  = {39}, number = {8}, pages = {895--935}, year = {2020},
  doi     = {10.1177/0278364920917446}
}

@inproceedings{zhang2025near,
  author    = {Zhang, Qixin and Wan, Zongqi and Yang, Yu and Shen, Li and Tao, Dacheng},
  title     = {Near-Optimal Online Learning for Multi-Agent Submodular
               Coordination: Tight Approximation and Communication Efficiency},
  booktitle = {International Conference on Learning Representations (ICLR)},
  year      = {2025}
}

@inproceedings{chiang2012,
  author    = {Chiang, Chao-Kai and Yang, Tianbao and Lee, Chia-Jung and
               Mahdavi, Mehrdad and Lu, Chi-Jen and Jin, Rong and Zhu, Shenghuo},
  title     = {Online optimization with gradual variations},
  booktitle = {Conference on Learning Theory (COLT)},
  year      = {2012}
}

@inproceedings{rakhlin2013,
  author    = {Rakhlin, Alexander and Sridharan, Karthik},
  title     = {Online learning with predictable sequences},
  booktitle = {Conference on Learning Theory (COLT)},
  year      = {2013}
}

@inproceedings{zzlz2022,
  author    = {Zhang, Mengxiao and Zhao, Peng and Luo, Haipeng and Zhou, Zhi-Hua},
  title     = {No-Regret Learning in Time-Varying Zero-Sum Games},
  booktitle = {Proceedings of the 39th International Conference on Machine Learning},
  series    = {Proceedings of Machine Learning Research},
  volume    = {162},
  pages     = {26772--26808},
  publisher = {PMLR},
  year      = {2022}
}

@article{zhao2024,
  author  = {Zhao, Peng and Zhang, Yu-Jie and Zhang, Lijun and Zhou, Zhi-Hua},
  title   = {Adaptivity and Non-stationarity: Problem-dependent Dynamic Regret for Online               Convex Optimization},
  journal = {Journal of Machine Learning Research},
  volume  = {25},
  number  = {98},
  pages   = {1--52},
  year    = {2024}
}

@inproceedings{xu2026distributed,
  author    = {Zirui Xu and Vasileios Tzoumas},
  title     = {Distributed Online Submodular Maximization under Communication Delays: A Simultaneous Decision-Making Approach},
  booktitle = {American Control Conference (ACC)},
  pages     = {4105--4110},
  year      = {2026}
}

@article{xu2025communication,
  title={Communication- and Computation-Efficient Distributed Submodular Optimization in Robot Mesh Networks},
  author={Xu, Zirui and Garimella, Sandilya Sai and Tzoumas, Vasileios},
  journal={IEEE Transactions on Robotics (TRO)},
  year={2025}
}

@article{conforti1984submodular,
  title={Submodular set functions, matroids and the greedy algorithm: Tight worst-case bounds and some generalizations of the Rado-Edmonds theorem},
  author={Conforti, Michele and Cornu{\'e}jols, G{\'e}rard},
  journal={Discrete Applied Mathematics},
  volume={7},
  number={3},
  pages={251--274},
  year={1984},
  publisher={Elsevier}
}

@inproceedings{xu2024leveraging,
  title={Leveraging untrustworthy commands for multi-robot coordination in unpredictable environments: A bandit submodular maximization approach},
  author={Xu, Zirui and Lin, Xiaofeng and Tzoumas, Vasileios},
  booktitle={American Control Conference (ACC)},
  pages={192--199},
  year={2024}
}

@inproceedings{xu2023bandit,
  title={Bandit Submodular Maximization for Multi-Robot Coordination in Unpredictable and Partially Observable Environments},
  author={Xu, Zirui and Lin, Xiaofeng and Tzoumas, Vasileios},
  booktitle={Robotics: Science and Systems (RSS)},
  year={2023}
}

@article{xu2023online,
  title={Online Submodular Coordination With Bounded Tracking Regret: Theory, Algorithm, and Applications to Multi-Robot Coordination},
  author={Xu, Zirui and Zhou, Hongyu and Tzoumas, Vasileios},
  journal={IEEE Robotics and Automation Letters (RAL)},
  volume={8},
  number={4},
  pages={2261--2268},
  year={2023}
}

@article{herbster1998tracking,
  title={Tracking the best expert},
  author={Herbster, Mark and Warmuth, Manfred K},
  journal={Machine learning},
  volume={32},
  number={2},
  pages={151--178},
  year={1998},
  publisher={Springer}
}

@article{streeter2008online,
  title={An online algorithm for maximizing submodular functions},
  author={Streeter, Matthew and Golovin, Daniel},
  journal={Adv. Neu. Inf. Proc. Sys.},
  volume={21},
  year={2008}
}

@inproceedings{matsuoka2021tracking,
  title={Tracking Regret Bounds for Online Submodular Optimization},
  author={Matsuoka, Tatsuya and Ito, Shinji and Ohsaka, Naoto},
  booktitle={International Conference on Artificial Intelligence and Statistics (AISTATS)},
  pages={3421--3429},
  year={2021},
  organization={PMLR}
}

@article{gharesifard2017distributed,
  title={Distributed submodular maximization with limited information},
  author={Gharesifard, Bahman and Smith, Stephen L},
  journal={IEEE Transactions on Control of Network Systems (TCNS)},
  volume={5},
  number={4},
  pages={1635--1645},
  year={2017},
  publisher={IEEE}
}

@article{rezazadeh2023distributed,
  title={Distributed strategy selection: A submodular set function maximization approach},
  author={Rezazadeh, Navid and Kia, Solmaz S},
  journal={Automatica},
  volume={153},
  pages={111000},
  year={2023},
  publisher={Elsevier}
}

@inproceedings{robey2021optimal,
  title={Optimal algorithms for submodular maximization with distributed constraints},
  author={Robey, Alexander and Adibi, Arman and Schlotfeldt, Brent and Hassani, Hamed and Pappas, George J},
  booktitle={Learn. for Dyn. \& Cont. (L4DC)},
  pages={150--162},
  year={2021},
}

@article{schlotfeldt2021resilient,
  title={Resilient active information acquisition with teams of robots},
  author={Schlotfeldt, Brent and Tzoumas, Vasileios and Pappas, George J},
  journal={IEEE Transactions on Robotics (TRO)},
  volume={38},
  number={1},
  pages={244--261},
  year={2021},
  publisher={IEEE}
}

@article{liu2021distributed,
  title={Distributed resilient submodular action selection in adversarial environments},
  author={Liu, Jun and Zhou, Lifeng and Tokekar, Pratap and Williams, Ryan K},
  journal={IEEE Robotics and Automation Letters (RAL)},
  volume={6},
  number={3},
  pages={5832--5839},
  year={2021},
  publisher={IEEE}
}

@article{corah2019distributed,
  title={Distributed matroid-constrained submodular maximization for multi-robot exploration: Theory and practice},
  author={Corah, Micah and Michael, Nathan},
  journal={Autonomous Robots (AURO)},
  volume={43},
  number={2},
  pages={485--501},
  year={2019},
}

@inproceedings{corah2018distributed,
  title={Distributed submodular maximization on partition matroids for planning on large sensor networks},
  author={Corah, Micah and Michael, Nathan},
  booktitle={IEEE Conference on Decision and Control (CDC)},
  pages={6792--6799},
  year={2018},
}

@inproceedings{konda2022execution,
  title={Execution order matters in greedy algorithms with limited information},
  author={Konda, Rohit and Grimsman, David and Marden, Jason R},
  booktitle={American Control Conference (ACC)},
  pages={1305--1310},
  year={2022}
}

@inproceedings{chen2018online,
  title={Online Continuous Submodular Maximization},
  author={Chen, Lin and Hassani, Hamed and Karbasi, Amin},
  booktitle={International Conference on Artificial Intelligence and Statistics (AISTATS)},
  pages={1896--1905},
  year={2018},
  organization={PMLR}
}

@article{sviridenko2017optimal,
  title={Optimal approximation for submodular and supermodular optimization with bounded curvature},
  author={Sviridenko, Maxim and Vondr{\'a}k, Jan and Ward, Justin},
  journal={Math.~of Operations Research},
  volume={42},
  number={4},
  pages={1197--1218},
  year={2017},
}

@inproceedings{tokekar2014multi,
  title={Multi-target visual tracking with aerial robots},
  author={Tokekar, Pratap and Isler, Volkan and Franchi, Antonio},
  booktitle={IEEE/RSJ International Conference on Intelligent Robots and Systems (IROS)},
  pages={3067--3072},
  year={2014},
}

@incollection{fisher1978analysis,
  title={An analysis of approximations for maximizing submodular set functions--{II}},
  author={Fisher, Marshall L and Nemhauser, George L and Wolsey, Laurence A},
  booktitle={Polyhedral combinatorics},
  pages={73--87},
  year={1978},
}

@article{singh2009efficient,
  title={Efficient informative sensing using multiple robots},
  author={Singh, Amarjeet and Krause, Andreas and Guestrin, Carlos and Kaiser, William J},
  journal={Journal of Artificial Intelligence Research (JAIR)},
  volume={34},
  pages={707--755},
  year={2009}
}

@incollection{crama1989characterization,
  title={A characterization of a cone of pseudo-Boolean functions via supermodularity-type inequalities},
  author={Crama, Yves and Hammer, Peter L and Holzman, Ron},
  booktitle={Quantitative Methoden in den Wirtschaftswissenschaften},
  pages={53--55},
  year={1989},
  publisher={Springer}
}

@article{foldes2005submodularity,
  title={Submodularity, supermodularity, and higher-order monotonicities of pseudo-Boolean functions},
  author={Foldes, Stephan and Hammer, Peter L},
  journal={Mathematics of Operations Research},
  volume={30},
  number={2},
  pages={453--461},
  year={2005},
  publisher={INFORMS}
}

@inproceedings{atanasov2015decentralized,
  title={Decentralized active information acquisition: {T}heory and application to multi-robot {SLAM}},
  author={Atanasov, Nikolay and Le Ny, Jerome and Daniilidis, Kostas and Pappas, George J},
  booktitle={IEEE Inter. Conf. Rob. Auto. (ICRA)},
  pages={4775--4782},
  year={2015},
}

@book{cesa2006prediction,
  title={Prediction, Learning, and Games},
  author={Cesa-Bianchi, Nicolo and Lugosi, G{\'a}bor},
  year={2006},
  publisher={Cambridge University Press}
}

@inproceedings{corah2021scalable,
  title={Scalable Distributed Planning for Multi-Robot, Multi-Target Tracking},
  author={Corah, Micah and Michael, Nathan},
  booktitle={IEEE/RSJ International Conference on Intelligent Robots and Systems (IROS)},
  pages={437--444},
  year={2021}
}

@article{kalman1960new,
  title={A new approach to linear filtering and prediction problems},
  author={Kalman, Rudolph Emil},
  journal={Journal of Fluids Engineering},
  volume={82},
  number={1},
  pages={35--45},
  year={1960},
  publisher={American Society of Mechanical Engineers}
}

@article{krause2008near,
  title={Near-optimal sensor placements in Gaussian processes: Theory, efficient algorithms and empirical studies},
  author={Krause, Andreas and Singh, Ajit and Guestrin, Carlos},
  journal={Jour. Mach. Learn. Res. (JMLR)},
  volume={9},
  pages={235--284},
  year={2008},
}

@article{krause2012submodular,
  title={Submodular function maximization},
  author={Krause, Andreas and Golovin, Daniel},
  journal={Tractability: Practical Approaches to Hard Problems},
  volume={3},
  year={2012},
}

@article{Feige:1998:TLN:285055.285059,
 author = {Feige, Uriel},
 title = {A Threshold of $ln(n)$ for Approximating Set Cover},
 journal = {Journal of the ACM (JACM)},
 volume = {45},
 number = {4},
 year = {1998},
 pages = {634--652},
}

@article{grimsman2019impact,
  title={The impact of information in distributed submodular maximization},
  author={Grimsman, David and Ali, Mohd Shabbir and Hespanha, Joao P and Marden, Jason R},
  journal={IEEE Trans. Ctrl. Netw. Sys. (TCNS)},
  volume={6},
  number={4},
  pages={1334--1343},
  year={2019},
  publisher={IEEE}
}

\appendices

\section{Proof of \Cref{thm:ofsf}}
\label{app:ofsf}
 
We prove the precise form
\begin{align}
  \sum_{t=1}^{T} r_{i^\star_t,t} - \sum_{t=1}^{T} r_t^{\top} p_t & \leq \tfrac{3}{\sqrt{2}}\sqrt{\alpha_{\mathcal{V}}\,(1+2P(T))\,\bar D(T)} \notag \\&\quad+ 10\sqrt{\bar D(T)\ln J},
  \label{eq:thmofsf-precise}
\end{align}
with $\alpha_{\mathcal{V}} = 3+\ln(|\mathcal{V}|T)$ and
$J = \lceil\log_2(2T)\rceil+1$. \Cref{thm:ofsf} follows, since $\alpha_{\mathcal{V}} = O(\log(|\mathcal{V}|T))$, $\ln J = O(\log\log T)$, and $1+2P(T) \le 2(1+P(T))$.
 
Write $\lambda_{j,t} \triangleq r_t^{\top}p^{(j)}_t$ and $\mu_{j,t} \triangleq \hat r_t^{\top}p^{(j)}_t$ for the realized and predicted rewards of base expert $j$, and $B(T) \triangleq \sum_{t=1}^{T}\|\lambda_t-\mu_t\|_\infty^2$. The imported lemmas are stated for losses. We apply them to $\ell_t=-r_t$ and $m_t=-\hat r_t$, which leaves every norm unchanged and reverses the sign of the regret. After normalization, lines~8 and 14--16 of \Cref{alg:ofsf} are the base updates of \cite[eq.~(20)]{zzlz2022} with share coefficient $\xi=\beta$, and lines 9 and 17 are Optimistic Hedge \cite{zhao2024}.
 
\begin{lemma}[\hspace{-.1mm}{\cite[Lemma~13]{zzlz2022}}]\label{lem:base}
For any $j$ and any comparators $u_t\in\Delta_{|\mathcal{V}|}$ with $P^u \triangleq \sum_{t\ge2}\|u_t-u_{t-1}\|_1$, $\sum_{t=1}^{T}\langle u_t-p^{(j)}_t,\,r_t\rangle \le \alpha_{\mathcal{V}}(1+P^u)/\gamma^{(j)}+\gamma^{(j)}D(T)$.
\end{lemma}
 
\begin{lemma}[\hspace{-.1mm}{\cite[Lemma~2]{zhao2024}}]\label{lem:meta}
For Optimistic Hedge over $J$ experts at a fixed rate $\varepsilon>0$ and any expert $j$, $\sum_{t}\lambda_{j,t}-\sum_{t}\langle q_t,\lambda_t\rangle \le \varepsilon\sum_t\|\lambda_t-\mu_t\|_\infty^2+\ln(J)/\varepsilon$.
\end{lemma}
 
Both are stated with an additional non-positive term, dropped here. The remaining lemmas are ours.
 
\begin{lemma}\label{lem:budget}
$\|\lambda_t-\mu_t\|_\infty\le\|r_t-\hat r_t\|_\infty\le1$, hence $B(T)\le D(T)$.
\end{lemma}
\begin{proof}
$|\lambda_{j,t}-\mu_{j,t}|=|(r_t-\hat r_t)^{\top}p^{(j)}_t| \le\|r_t-\hat r_t\|_\infty$, since $\|p^{(j)}_t\|_1=1$ and $r_t,\hat r_t\in[0,1]^{|\mathcal{V}|}$.
\end{proof}
 
\begin{lemma}\label{lem:grid}
There exists $j$ such that $\alpha_{\mathcal{V}}(1+2P(T))/\gamma^{(j)}+\gamma^{(j)}\bar D(T) \le \tfrac{3}{\sqrt2}\sqrt{\alpha_{\mathcal{V}}(1+2P(T))\,\bar D(T)}$.
\end{lemma}
\begin{proof}
Let $A=\alpha_{\mathcal{V}}(1+2P(T))$ and $\gamma^\star=\sqrt{A/\bar D(T)}$. Since $A\in[\alpha_{\mathcal{V}},\alpha_{\mathcal{V}}(1+2T)]$ and $\bar D(T)\in[1,T]$, $\gamma^\star/\gamma^{(1)}\in[1,\sqrt{T(1+2T)}]\subseteq[1,2T]$. As $\gamma^{(j)}=\gamma^{(1)}2^{j-1}$ for $j\le J=\lceil\log_2(2T)\rceil+1$, some $j$ has $\gamma^{(j)}=\lambda\gamma^\star$ with $\lambda\in[2^{-1/2},2^{1/2}]$, giving $A/\gamma^{(j)}+\gamma^{(j)}\bar D(T)=(\lambda+\lambda^{-1})\sqrt{A\bar D(T)} \le(2^{1/2}+2^{-1/2})\sqrt{A\bar D(T)}=\tfrac{3}{\sqrt2}\sqrt{A\bar D(T)}$.
\end{proof}
 
\begin{lemma}\label{lem:restart}
For every $j$, $\sum_t\lambda_{j,t}-\sum_t\langle q_t,\lambda_t\rangle \le 10\sqrt{\bar D(T)\ln J}$.
\end{lemma}
\begin{proof}
We use a variant of the doubling trick~\cite{cesa2006prediction}, restarting on the observed prediction-error budget rather than on fixed rounds, as suggested in~\cite{zhao2024}.

Epoch $\ell$ restarts the meta at the fixed rate $\varepsilon_\ell=\sqrt{\ln(J)/4^\ell}$, so \Cref{lem:meta} applies to it separately. Its budget $b_\ell$ satisfies $b_\ell\le4^\ell+1$: $B_t\le4^\ell$ before its last round, which adds at most $1$ by \Cref{lem:budget}. The epoch's regret is thus at most $\varepsilon_\ell(4^\ell+1)+\ln(J)/\varepsilon_\ell =(2\cdot2^\ell+2^{-\ell})\sqrt{\ln J}$. If $L$ is the last epoch, $4^{L-1}<\bar B(T)$ for $L\ge1$, so $2^L\le2\sqrt{\bar B(T)}$ in all cases. Summing, $\sum_{\ell\le L}(2\cdot2^\ell+2^{-\ell})\le 4\cdot2^L+2 \le8\sqrt{\bar B(T)}+2\le10\sqrt{\bar B(T)}$, and $\bar B(T)\le\bar D(T)$ by \Cref{lem:budget}.
\end{proof}
 
\begin{proof}[Proof of \Cref{thm:ofsf}]
Since $p_t=\sum_jq_{j,t}p^{(j)}_t$, we have $r_t^{\top}p_t=\langle q_t,\lambda_t\rangle$. Hence, for any $j$,
\begin{align*}
\sum_t r_{i^\star_t,t}-\sum_t r_t^{\top}p_t &=\Big[\sum_t\lambda_{j,t}-\sum_t\langle q_t,\lambda_t\rangle\Big] \\&\quad +\Big[\sum_t r_{i^\star_t,t}-\sum_t r_t^{\top}p^{(j)}_t\Big].
\end{align*}
The first bracket is at most $10\sqrt{\bar D(T)\ln J}$ by \Cref{lem:restart}. For the second, take $u_t=e_{i^\star_t}$ in \Cref{lem:base}. Since $\|e_i-e_{i'}\|_1=2$ for $i\ne i'$, $P^u=2P(T)$, and with $D(T)\le\bar D(T)$ the bracket is at most $A/\gamma^{(j)}+\gamma^{(j)}\bar D(T)$. Choosing $j$ as in \Cref{lem:grid} gives $\tfrac{3}{\sqrt2}\sqrt{\alpha_{\mathcal{V}}(1+2P(T))\,\bar D(T)}$. Every step is an inequality between realized quantities, so the bound holds on every realization.
\end{proof}

\end{document}